\documentclass[a4paper,11pt]{article}

\usepackage{a4wide}
\usepackage{amssymb}
\usepackage{amsmath}
\usepackage{amsthm}

\usepackage{dpda-macros}

\title{Bisimulation equivalence  
	of first-order grammars is Ackermann-hard}

\author{Petr Jan\v{c}ar\\
 \ \\
Dept Comp. Sci., FEI,	Techn. Univ. Ostrava, Czech Republic,
 \\
 http://www.cs.vsb.cz/jancar/,
 \\
 email: petr.jancar@vsb.cz
}

\date{}

\begin{document}

\maketitle

\begin{abstract}
\noindent	
Bisimulation equivalence 
(or bisimilarity)
of
first-order grammars is decidable, as follows from  
the decidability result 
by S\'enizergues (1998, 2005) that has been given
in an equivalent framework of 
equational graphs with finite out-degree, or of 
pushdown automata (PDA)
with
only deterministic and popping $\varepsilon$-transitions.
Benedikt, G\"oller, Kiefer, and Murawski (2013) have shown that the
bisimilarity problem for PDA (even) without $\varepsilon$-transitions is
nonelementary. 
Here we show Ackermann-hardness for bisimilarity of first-order grammars. 
The grammars do not use explicit $\varepsilon$-transitions, but they
correspond to the above mentioned PDA with (deterministic and popping)  
$\varepsilon$-transitions, and this feature is substantial in the
presented lower-bound proof.
The proof is based on
a (polynomial) reduction from the reachability problem of reset (or
lossy) counter machines, for which  the Ackermann-hardness has been
shown by Schnoebelen (2010); in fact, this reachability problem is known to be 
Ackermann-complete, i.e.,
$\mathbf{F}_\omega$-complete in the hierarchy of fast-growing
complexity classes defined by Schmitz (2013).
\end{abstract}

\noindent
\textbf{Keywords:}
\\
first-order grammar, term rewriting, pushdown
automaton, bisimulation, complexity

\subsection*{Introduction}

Bisimulation equivalence, 
also called bisimilarity, has been recognized as a fundamental behavioural equivalence 
of systems. It is thus natural to explore the related decidability and
complexity questions for various computational models. 

The involved result by S\'enizergues~\cite{Seni05}, showing the decidability of
bisimilarity for equational graphs of finite out-degree, or
equivalently for 
pushdown automata (PDA)
with
only deterministic and popping $\varepsilon$-transitions, surely
belongs to the most fundamental results in this area. (This result
generalized S\'enizergues's solving of the famous DPDA equivalence
problem.)
The complexity of the bisimilarity problem has been recently shown to
be nonelementary, by
Benedikt, G\"oller, Kiefer, and Murawski~\cite{BGKM12}, even for PDA
with no $\varepsilon$-transitions.

Here we look at the bisimilarity problem in the framework of 
first-order grammars (i.e., finite sets of term root-rewriting rules). 
This framework is long known as equivalent to the PDA framework.
(We can refer, e.g., to~\cite{JancarLICS12} for a recent use of a concrete
respective transformation.)
Hence  S\'enizergues's decidability proof applies to the grammars as
well.
We show that the problem for grammars is
``Ackermann-hard'', and thus not primitive recursive. 
Though the grammars do not use explicit $\varepsilon$-transitions, they
correspond to the above mentioned PDA with deterministic and popping
$\varepsilon$-transitions, and this feature is substantial in the
presented lower-bound proof.
Hence the nonelementary bound of~\cite{BGKM12} remains
the best known lower bound for bisimilarity of pushdown systems
without $\varepsilon$-transitions.

The presented proof is based on
a (polynomial) reduction from the reachability problem of reset (or
lossy) counter machines, for which  the Ackermann-hardness has been
shown by Schnoebelen (see~\cite{DBLP:conf/mfcs/Schnoebelen10} and the
reference therein).
In fact, we also know an ``Ackermannian'' upper bound for this
reachability problem~\cite{DBLP:conf/lics/FigueiraFSS11}, 
and the problem is thus $\mathbf{F}_\omega$-complete (or 
$\textsc{Ack}$-complete) 
 in the hierarchy of fast-growing
complexity classes defined by Schmitz~\cite{Schmitz2013}.
The question of a similar
upper bound for the bisimilarity problem is not addressed in this
paper.

\subsection*{Definitions, and the result}

We briefly recall the notions 
that are needed for stating the result. We use the forms that are
convenient here; e.g., we (harmlessly) 
define bisimulations as symmetric, though
this is usually not required. 

\textbf{First-order terms.}
We assume a fixed set of 
\emph{variables} $\var=\{x_1,x_2,\dots\}$.
Given a set $\calN$ of function symbols with arities,
by $\trees_\calN$ we
denote the set of \emph{terms over} $\calN$.
A \emph{term} $T\in\trees_\calN$
is either  a variable $x_i$ 
or  $A(U_1,\dots,U_m)$ where 
$A\in\calN$, 
$\arity(A)=m$, and $U_1,\dots,U_m$ are terms.

\textbf{First-order grammars.}
A (first-order) \emph{grammar} is a tuple $\calG=(\calN,\act,\calR)$
where $\calN$ is a finite set of ranked \emph{nonterminals}
(or function symbols with arities),
$\act$ is a finite set of  \emph{actions} (or terminals),
and $\calR$ is a finite set of (root-rewriting) \emph{rules} of the form
$A(x_1,\dots,x_m)\gt{a}V$ where 
$A\in\calN$, $\arity(A)=m$,
$a\in\act$, and
$V\in\trees_{\calN}$ is a 
term
in which all occurring variables are from the set 
$\{x_1,\dots,x_m\}$.
(Some example rules are $A(x_1,x_2,x_3)\gt{b}C(D(x_3,B),x_2)$,
$A(x_1,x_2,x_3)\gt{b}x_2$,
$D(x_1,x_2)\gt{a}A(D(x_2,x_2),x_1,B)$; here
 the arities of $A,B,C,D$ are $3,0,2,2$, respectively.)

\textbf{LTSs associated with grammars.}
A grammar $\calG=(\calN,\act,\calR)$ defines the \emph{labelled
transition system} $\calL_\calG=(\trees_\calN,\act,(\gt{a})_{a\in\act})$
in which each rule $A(x_1,\dots,x_m)\gt{a}V$ from $\calR$ induces 
transitions 
$(A(x_1,\dots,x_m))\sigma\gt{a}V\sigma$ for all substitutions
$\sigma:\var\rightarrow\trees_\calN$.
(The above example rules thus induce, e.g., 
$A(x_1,x_2,x_3)\gt{b}C(D(x_3,B),x_2)$ (here $\sigma(x_i)=x_i$),
$A(V,x_5,U)\gt{b}C(D(U,B),x_5)$ (here
$\sigma(x_1)=V$, $\sigma(x_2)=x_5$, $\sigma(x_3)=U$),
$A(U_1,U_2,U_3)\gt{b}C(D(U_3,B),U_2)$, $A(U_1,U_2,U_3)\gt{b}U_2$,
etc.)

\textbf{Bisimulation equivalence.}
Given $\calG=(\calN,\act,\calR)$, a set (or a relation)
$\calB\subseteq \trees_\calN\times\trees_\calN$ is a
\emph{bisimulation} if it is \emph{symmetric} ($(T,U)\in\calB$ implies
$(U,T)\in\calB$), and for any $(T,U)\in\calB$ and $T\gt{a}T'$ there is
$U'$ such that  $U\gt{a}U'$ and $(T',U')\in\calB$.
Two \emph{terms} $T,U$ are \emph{bisimilar}, written $T\sim U$, if there is a
bisimulation containing $(T,U)$.
\\
The \emph{bisimilarity problem for first-order grammars} asks, given a grammar $\calG$ and
terms $T,U$, whether $T\sim U$.

\textbf{Ackermann-hardness.}
We refer to~\cite{Schmitz2013} for detailed definitions of the class
 $\mathbf{F}_\omega$ (or $\textsc{Ack}$)
of decision problems, and of the problems that are 
$\mathbf{F}_\omega$-complete
or $\mathbf{F}_\omega$-hard.
Here we just recall the ``Ackermannian'' function
$f_A:\Nat\rightarrow\Nat$ (where $\Nat=\{0,1,2,\dots\}$) defined as
follows: we first define  the family $f_0,f_1,f_2,\dots$ 
by putting $f_0(n)=n{+}1$ and $f_{k+1}(n)=f_k(f_k( \dots f_k(n)\dots))$
where $f_k$ is applied $n{+}1$ times; then we put $f_A(n)=f_n(n)$.
The problem $\textsc{HP}_\textsc{Ack}$ that asks,
given a Turing machine $M$, an input $w$, and some
$n\in\Nat$, whether $M$ halts on $w$ within $f_A(n)$ steps, is an
example of an \emph{Ackermann-complete problem}. 
We say (in this paper) that a \emph{problem}
$\calP$ is
\emph{Ackermann-hard} if $\textsc{HP}_\textsc{Ack}$ is reducible 
to $\calP$, or to the complementary problem co-$\calP$, 
by a standard polynomial
many-one reduction.
(The notion is more general in~\cite{Schmitz2013}, and it also 
includes primitive-recursive reductions.)

\begin{theorem}\label{th:bisackhard}
The bisimilarity problem for
first-order grammars is Ackermann-hard. 
\end{theorem}

\subsection*{Proof of the theorem}

A direct reduction from $\textsc{HP}_\textsc{Ack}$ would require many
technicalities. Fortunately, these have been already handled in 
deriving  Ackermann-hardness (or Ackermann-completeness) of other
problems.
For our reduction we thus choose a more convenient
Ackermann-hard problem (which is Ackermann-complete, in
fact), namely the reachability problem for reset counter machines.

\textbf{Reset Counter Machines (RCMs).}
An \emph{RCM} is a tuple  $\calM=(d,Q,\delta)$ where 
$d$ is the \emph{dimension}, yielding $d$ nonnegative
\emph{counters} $c_1,c_2,\dots,c_d$,
$Q$ is a
finite set of \emph{(control) states},
and $\delta\subseteq Q\times\ACT\times Q$ is a finite set of
\emph{instructions}, where the set 
$\ACT$ of \emph{operations} contains $\incr(c_i)$ (increment $c_i$),
$\decr(c_i)$ (decrement $c_i$), 
and $\reset(c_i)$ (set $c_i$ to $0$),
for $i=1,2,\dots,d$.
We view $Q\times\Nat^d$ 
as the set $\conf$ of 
\emph{configurations} of $\calM$. 
The \emph{transition relation} $\gt{}\subseteq \conf\times\conf$ is
induced by $\delta$ in the obvious way:
If $(p,op,q)\in\delta$ then we have $(p,
(n_1,\dots,n_d))\gt{}(q,(n'_1,\dots,n'_d))$
in the following cases: 
\begin{itemize}
	\item
		$op=\incr(c_i)$, $n'_i=n_i{+}1$, and $n'_j=n_j$ for
		all $j\neq i$; or 
	\item
$op=\decr(c_i)$, $n_i>0$, $n'_i=n_i{-}1$, and $n'_j=n_j$ for
		all $j\neq i$; or 
	\item
$op=\reset(c_i)$, $n'_i=0$, and $n'_j=n_j$ for
		all $j\neq i$.
\end{itemize}		
By $\gt{}^*$ we denote the reflexive and transitive closure of
$\gt{}$.

\textbf{Complexity of the reachability problem for RCMs.}
We define the \emph{RCM-reachability problem} in the following
convenient form: given an RCM $\calM=(d,Q,\delta)$ and (control) states
$p_\init$, $p_\final$, we ask if $p_\final$ is reachable from
$(p_\init,(0,0,\dots,0))$, i.e., if there are $m_1,m_2,\dots,m_d\in\Nat$
such that
$(p_{\init},(0,0,\dots,0))\gt{}^*(p_\final,(m_1,m_2,\dots,m_d))$.
The known results yield:

\begin{theorem}\label{th:RCMAck}
RCM-reachability problem is Ackermann-complete.
\end{theorem}

It is the lower bound, the  Ackermann-hardness, which is important
for us here; we refer to~\cite{DBLP:conf/mfcs/Schnoebelen10} for a proof 
and further references.
(See also~\cite{DBLP:journals/jsyml/Urquhart99}
for an independent proof related to relevance logic.)
The upper bound follows from~\cite{DBLP:conf/lics/FigueiraFSS11}.
Hence the problem is
$\mathbf{F}_\omega$-complete in the sense of~\cite{Schmitz2013}.
We note that the same result holds for
\emph{lossy counter machines}; they have the zero-test instead of the
reset, and any counter can spontaneously decrease at any time.
(We prefer RCMs since they are a slightly simpler model.)

\textbf{RCM-reachability reduces to first-order bisimilarity.}
We finish by proving  the next lemma; this establishes 
Theorem~\ref{th:bisackhard}, by 
using the hardness part of Theorem~\ref{th:RCMAck}.
The reduction in the proof of the lemma is obviously polynomial; in
fact, it can be checked to be a logspace reduction, 
but this is a minor point in the view of the fact that even a 
primitive-recursive reduction would suffice here.

\begin{lemma}
The RCM-reachability problem is polynomially reducible to the complement 
of the bisimilarity problem for
first-order grammars.
\end{lemma}

\begin{proof}
Let us consider an instance $\calM=(d,Q,\delta)$,
$p_\init$, $p_\final$ of the RCM-reachability problem,
and imagine the
following game between Attacker (he) and Defender (she). This is the first
version of a game that will be afterwards
implemented as a standard bisimulation game.
Attacker aims to show that $p_\final$ is reachable from 
$(p_{\init}, (0,0,\dots,0))$, while 
Defender opposes this.

The game uses $2d$ \emph{game-counters}, which are never decremented;
each counter $c_i$ of $\calM$ yields two game-counters, namely 
$c_i^I$ and $c_i^D$, for counting the numbers of Increments and
Decrements of $c_i$, respectively, since the last
reset or since the beginning if there has been no reset of $c_i$ so
far.
The \emph{initial position} is 
$(p_\init, ((0,0), \dots,(0,0)))$, with all 
$2d$ game-counters (organized in pairs) having the value $0$.

A game \emph{round} from position $(p,((n_1,n'_1), \dots,
(n_d,n'_d)))$
proceeds as described below.
It will become clear that 
it suffices to consider only the cases $n_i\geq n'_i$;
the position then corresponds to the $\calM$'s configuration
$(p,(n_1{-}n'_1, \dots,  n_d{-}n'_d))$.

If $p=p_\final$, then Attacker wins;
if  $p\neq p_\final$ and there is no instruction $(p,op,q)\in\delta$, then 
Defender wins.
Otherwise Attacker chooses $(p,op,q)\in\delta$, and the continuation
depends on $op$:
\begin{enumerate}
\item
If $op=\incr(c_i)$, then the next-round position arises (from the
previous one) by replacing $p$ with $q$ and by performing 
$c^I_i:=c^I_i{+}1$ (the counter of increments of $c_i$ is
incremented, i.e., $n_i$ is replaced with $n_i{+}1$).
\item
If $op=\reset(c_i)$, then the next-round position arises 
by replacing $p$ with $q$ and by performing 
$c^I_i:=0$ and $c^D_i:=0$ (hence both $n_i$ and $n'_i$ are replaced
with $0$).
\item
If $op=\decr(c_i)$, then \emph{Defender chooses} one of the following
options:
\begin{enumerate}
\item
the next-round position arises 
by replacing $p$ with $q$ and by performing 
$c^D_i:=c^D_i{+}1$ (the counter of decrements of $c_i$ is
incremented, i.e., $n'_i$ is replaced with $n'_i{+}1$), or
\item
	(Defender claims that this decrement is \emph{illegal} since $n_i=n'_i$ and)
the next position becomes just $(n_i,n'_i)$. 
In this case a (deterministic) check if $n_i=n'_i$ 
is performed, by successive
synchronized decrements at both sides. If indeed
$n_i=n'_i$ (the counter-bottoms are reached at the same moment),
then Defender wins; otherwise (when $n_i\neq n'_i$)
Attacker wins.
\end{enumerate}
\end{enumerate}
If
$(p_\init,(0,0,\dots,0))\gt{}^*(p_\final,(m_1,m_2,\dots,m_d))$ for some
$m_1,m_2,\dots,m_d$, i.e., if the answer to RCM-reachability is YES,
then Attacker has a winning strategy: he just follows the
corresponding
sequence of instructions. He thus also always chooses $\decr(c_i)$
\emph{legally},
i.e. only in the cases where $n_i>n'_i$, and Defender loses if
she ever chooses 3(b).
If the answer is NO ($p_\final$ is not reachable), and Attacker
follows a legal sequence of instructions, then he
either loses in a ``dead'' state or the play is infinite; if Attacker
chooses an illegal decrement, then in the first such situation we
obviously have
$n_i=n'_i$ for the respective counter $c_i$, and Defender can force
her win via 3(b).

Since the game-counters can be only incremented or reset, it is a
routine
to implement the above game as a bisimulation game in the
grammar-framework
(using a standard
method of ``Defender's forcing'' for implementing the choice in 3).
We now describe the corresponding 
grammar $\calG=(\calN,\Sigma,\calR)$.

The set $\calN$ of nonterminals will include 
a unary nonterminal $I$, a nullary nonterminal $\bot$, and the
nonterminals with arity  $2d$ that are induced by control states of
$\calM$ as follows:
each $p\in Q$ induces 
$A_p, A_{(p,i)}, B_p, B_{(p,i,1)},
B_{(p,i,2)}$, where $i=1,2,\dots,d$.

We intend that 
a game-position $(p,((n_1,n'_1),\dots,(n_d,n'_d)))$ corresponds to the
pair of terms 
\begin{equation}\label{eq:pairposition}
\left(A_p(I^{n_1}\bot, I^{n'_1}\bot,\dots, I^{n_d}\bot,
I^{n'_d}\bot), \;B_p(I^{n_1}\bot, I^{n'_1}\bot,\dots, I^{n_d}\bot,
I^{n'_d}\bot)\right)
\end{equation}
where $I^k\bot$ is a shorthand for 
$I(I(\dots I(\bot)\dots))$ with $I$ occurring $k$ times; we put
$I^0\bot=\bot$.
The RCM-reachability instance $\calM, p_\init, p_\final$ will be
reduced to the (non)bisimilarity-problem instance $\calG$, 
$A_{p_\init}(\bot,\dots,\bot)$, $B_{p_\init}(\bot,\dots,\bot)$.

We put $\act=\delta\,\uplus\,\{a,b\}$, i.e., the actions of
$\calG$ correspond to the instructions (or instruction names)
of $\calM$, and we also use
auxiliary actions $a$, $b$.

The set of rules $\calR$ contains a sole rule for $I$, namely
$I(x_1)\gt{a}x_1$, and no rule for $\bot$;
hence $I^n\bot\sim I^{n'}\bot$ iff $n=n'$.
Each instruction   $\instr=(p,op,q)\in\delta$ induces the rules in $\calR$ as
follows:

\begin{enumerate}
		\item
If $op=\incr(c_i)$, then the induced rules are
\\
$A_p(x_1,\dots,x_{2d})\gt{\instr}A_q(x_1,\dots,x_{2(i-1)},
I(x_{2i-1}), x_{2i},\dots,x_{2d})$, and
\\
$B_p(x_1,\dots,x_{2d})\gt{\instr}B_q(x_1,\dots,x_{2(i-1)},
I(x_{2i-1}), x_{2i},\dots,x_{2d})$.
\item
If $op=\reset(c_i)$, then the induced rules are
\\
$A_p(x_1,\dots,x_{2d})\gt{\instr}A_q(x_1,\dots,x_{2(i-1)},
\bot,\bot,x_{2i+1},\dots,x_{2d})$,
\\
$B_p(x_1,\dots,x_{2d})\gt{\instr}B_q(x_1,\dots,x_{2(i-1)},
\bot,\bot,x_{2i+1},\dots,x_{2d})$.
\item
If $op=\decr(c_i)$, then the induced rules are below;
here we use the shorthand $A\gt{a}B$ when meaning 
$A(x_1,\dots,x_{2d})\gt{a}B(x_1,\dots,x_{2d})$:
\\
$A_p\gt{\instr}A_{(q,i)}$,
$A_p\gt{\instr}B_{(q,i,1)}$,
$A_p\gt{\instr}B_{(q,i,2)}$,
$B_p\gt{\instr}B_{(q,i,1)}$,
$B_p\gt{\instr}B_{(q,i,2)}$,
\\
$A_{(q,i)}(x_1,\dots,x_{2d})\gt{a}A_q(x_1,\dots,x_{2i-1},I(x_{2i}),
x_{2i+1},\dots,x_{2d})$,
\\
$B_{(q,i,1)}(x_1,\dots,x_{2d})\gt{a}B_q(x_1,\dots,x_{2i-1},I(x_{2i}),
x_{2i+1},\dots,x_{2d})$,
\\
$B_{(q,i,2)}(x_1,\dots,x_{2d})\gt{a}A_q(x_1,\dots,x_{2i-1},I(x_{2i}),
x_{2i+1},\dots,x_{2d})$,
\\
$A_{(q,i)}(x_1,\dots,x_{2d})\gt{b}x_{2i-1}$,
$B_{(q,i,1)}(x_1,\dots,x_{2d})\gt{b}x_{2i-1}$,
$B_{(q,i,2)}(x_1,\dots,x_{2d})\gt{b}x_{2i}$.
\end{enumerate}
Moreover, $\calR$ will also contain
$A_{p_\final}(x_1,\dots,x_{2d})\gt{a}\bot$ 
(but \emph{not} $B_{p_\final}(x_1,\dots,x_{2d})\gt{a}\bot$). 

Now we recall the standard (turn-based) \emph{bisimulation game}, starting
with the pair
$(A_{p_\init}(\bot,\dots,\bot),B_{p_\init}(\bot,\dots,\bot))$.
In the round starting with $(T_1,T_2)$, Attacker chooses a transition
$T_j\gt{a}T'_j$ and then Defender chooses $T_{3-j}\gt{a}T'_{3-j}$ (for
the same $a\in\act$); the
next round starts with the pair $(T'_1,T'_2)$. If a player gets
stuck, then (s)he loses; an infinite play is a win of Defender.
It is obvious that Defender has a winning strategy in this game iff 
$A_{p_\init}(\bot,\dots,\bot)\sim B_{p_\init}(\bot,\dots,\bot)$.

We now easily check that this bisimulation game indeed implements the
above described game; a game-position  
$(p,((n_1,n'_1),\dots,(n_d,n'_d)))$ is implemented as the
pair~(\ref{eq:pairposition}). The points 1 and 2 directly correspond
to the previous points 1 and 2.
If Attacker chooses an instruction $\instr=(p,\decr(c_i),q)$, 
then he
must use the respective rule $A_p\gt{\instr}A_{(q,i)}$ in 3, since
otherwise Defender installs syntactic equality, i.e. a pair $(T,T)$.
It is now Defender who chooses $B_p\gt{\instr}B_{(q,i,1)}$ (corresponding
to the previous 3(a)) or $B_p\gt{\instr}B_{(q,i,2)}$ (corresponding
to 3(b)). Attacker then must
choose action $a$ in the first case, and action $b$ in the second
case; otherwise we get syntactic equality.
The first case thus results in the pair 
$(A_q(\dots), B_q(\dots))$ corresponding to the next game-position
(where $c_i^D$ has been incremented), and the second case results in
the pair $(I^{n_i}\bot, I^{n'_i}\bot)$; we have already observed that 
$I^{n_i}\bot\sim I^{n'_i}\bot$ iff $n_i=n'_i$.

Finally we observe that in any pair 
$(A_{p_{\final}}(\dots), B_{p_\final}(\dots))$ Attacker wins immediately
(since the transition $A_{p_\final}(\dots)\gt{a}\bot$ can not be
matched).

We have thus established that $p_\final$ is reachable from
$(p_\init, (0,\dots,0))$ if, and only if, 
$A_{p_\init}(\bot,\dots,\bot)\not\sim
B_{p_\init}(\bot,\dots,\bot)$.
\end{proof}

\subsection*{Additional remarks}

In the above bisimulation game we obviously encounter ``unbalanced'' terms; 
the syntactic tree of an unbalanced term has branches of different lengths. 
This feature is related to deterministic popping
$\varepsilon$-transitions in the pushdown automata corresponding to
our grammars.
Hence the nonelementary bound shown in~\cite{BGKM12} remains
the best known lower bound for bisimilarity of pushdown systems
without $\varepsilon$-transitions.

We can add that introducing (popping) $\varepsilon$-rules
$A(x_1,\dots,x_m)\gt{\varepsilon}x_i$ in our grammars (where such a
rule might be not the only one with $A(x_1,\dots,x_m)$ at the
left-hand side) already yields undecidability of
bisimilarity~\cite{DBLP:journals/jacm/JancarS08}.

\subsection*{Author's acknowledgement}

The reported result has been achieved during my visit at LSV ENS
Cachan, and 
I am grateful to Sylvain Schmitz and Philippe Schnoebelen for
fruitful discussions.

\bibliographystyle{abbrv}
\bibliography{root}

\end{document}